\documentclass[journal,draft,onecolumn,12pt]{IEEEtran}
\usepackage{amsfonts,color,morefloats}
\usepackage{amssymb,amsmath,latexsym,amsthm}

\newcommand{\gf}{{\mathrm{GF}}}

\newcommand{\C}{{\mathcal{C}}}

\newcommand{\bD}{{\mathcal{D}}}

\newcommand{\bu}{{\mathbf{u}}}

\newcommand{\bone}{{\mathbf{1}}}

\newtheorem{theorem}{Theorem}
\newtheorem{lemma}[theorem]{Lemma}

\newtheorem{corollary}[theorem]{Corollary}

\newtheorem{problem}{Open Problem}
\newtheorem{definition}{Definition}
\newtheorem{example}{Example}

\begin{document}

\title{More  MDS codes of non-Reed-Solomon type  \thanks{
Y. Wu's research was sponsored by the National Natural Science Foundation of China under Grant  12101326 and 62372247, the
Natural Science Foundation of Jiangsu Province under Grant BK20210575, and the China Postdoctoral Science Foundation under Grant 2023M740958. Z. Heng was sponsored by the National Natural Science Foundation of China under Grant 12271059 and 11901049, and in part by the Fundamental Research Funds for the Central Universities, CHD, under Grant 300102122202. C. Li was supported by the National Natural Science Foundation of China (Nos. T2322007,
12071138), Shanghai Rising-Star Program (No. 22QA1403200), and “Digital
Silk Road” Shanghai International Joint Lab of Trustworthy Intelligent
Software (No. 22510750100). C. Ding's research was supported by The Hong Kong Research Grants Council, Proj. No. $16301523$ and in part by 
the UAEU-AUA joint research grant G00004614.
}
}

\author{Yansheng Wu,\thanks{Y. Wu is with the School of Computer Science, Nanjing University of Posts and Telecommunications, Nanjing,
210023, P. R. China. Email: 
yanshengwu@njupt.edu.cn.} 
 Ziling Heng, \thanks{Z. Heng is with the School of Science, Chang’an University, Xi’an, 710064, P. R. China. Email: zilingheng@chd.edu.cn.}
\and Chengju Li, \thanks{C. Li is with the Shanghai Key Laboratory of Trustworthy Computing, East China Normal University, Shanghai, 200062, P. R. China. Email: cjli@sei.ecnu.edu.cn.}
\and Cunsheng Ding\thanks{C. Ding is with the Department of Computer Science
                           and Engineering, The Hong Kong University of Science and Technology,
Clear Water Bay, Kowloon, Hong Kong, China. Email: cding@ust.hk}
}




\date{\today}
\maketitle




\begin{abstract}

MDS codes have  diverse practical applications in communication systems, data storage,  and quantum codes due to their algebraic properties and optimal error-correcting capability. 
In this paper, we focus on a  class of linear codes and establish some sufficient and necessary conditions for them being MDS. Notably,  these codes differ from Reed-Solomon codes up to monomial equivalence.  Additionally, we also explore the cases in which these codes are almost MDS or near MDS.
Applying our main results, we determine the covering radii and deep holes of  the dual codes associated with specific Roth-Lempel codes and   discover an infinite family of (almost) optimally extendable codes with dimension  three.


\end{abstract}


\section{Introduction}

Let $\gf(q)$ be the finite field of order $q$, where $q$ is a power of a prime number.
An ${[n, k,d]}$ linear code $\mathcal{C}$ over $\gf(q)$  is a $k$-dimensional
subspace of $\gf(q)^n$ with the minimum Hamming distance $d$.  
The dual code of a linear code $\C$ of length $n$ over $\gf(q)$ is defined by
$$\mathcal C^{\bot} = \left\{ (x_1, \ldots, x_{n})={\bf x}\in \gf(q)^{n} \mid \langle {\bf x},{\bf y}\rangle=\sum_{i=1}^{n}x_iy_i=0~  \forall~  {\bf y}= (y_1, \ldots,y_{n})\in \mathcal C\right\}.$$

The minimum distance $d$ is bounded  by the so-called {\em Singleton bound}, that is, $d\le n-k+1$.
If $d= n-k+1$, then the code $\mathcal{C}$ is called a {\em maximum distance separable} (MDS) code. If $d=n-k$, then the code $\mathcal{C}$ is said to be almost MDS (AMDS). Additively, $\mathcal{C}$ is said to be  near MDS (NMDS) if the code and its dual are both AMDS.

An MDS code has the optimal error correcting capability
when its length and dimension are fixed. MDS codes are extensively used in communications (for example,
 Reed-Solomon codes are MDS codes), and they have good applications in minimum storage codes and quantum codes.  There are many known constructions for MDS codes; for instance, {\em Generalized Reed-Solomon} (GRS) codes and extended GRS codes
 \cite{RS}, and  codes based on $n$-arcs in projective geometry \cite{MS}, circulant matrices \cite{RL},  and Hankel matrices \cite{RS2}.
  
 Roth and Lempel \cite{RL} introduced a novel construction for  MDS codes. Their codes $\mathcal{C}_1$ 
 over $\gf(q)$ have the following generator matrix: 
 \begin{equation}\label{eq:1.1} G_1=
\left( \begin{array}{cccccc}
1 & 1& \ldots &  1 & 0 &0\\
\alpha_1& \alpha_2& \ldots &  \alpha_{n} & 0 &0\\
\vdots& \vdots& \ldots & \vdots & \vdots &\vdots\\
\alpha_1^{k-3}& \alpha_2^{k-3}& \ldots &  \alpha_{n}^{k-3} & 0 &0\\
\alpha_1^{k-2}& \alpha_2^{k-2}& \ldots &  \alpha_{n}^{k-2} & 0 &1\\
\alpha_1^{k-1}& \alpha_2^{k-1}& \ldots &  \alpha_{n}^{k-1} & 1 &\delta\\
\end{array} \right),
\end{equation} where $\alpha_{1},\ldots,\alpha_{n}$ are distinct elements in $\gf(q)$ and $ \delta\in \gf(q)$, and $4\le k + 1 \le n \le q$.
 They  proved that the code $\C_1$ is  MDS  if and only if the set $\{\alpha_1, \ldots, \alpha_{n}\}$ is an $(n, k-1, \delta)$-set in $\gf(q)$.  Wu, Hyun, and Lee \cite{WHL} utilized the Roth-Lempel codes to construct LCD MDS codes. Their construction leveraged the properties of the Roth-Lempel codes to achieve the LCD requirement.
In a separate study, Han and Fan \cite{HF} obtained near MDS codes of non-elliptic-curve type. Their research focused on finding MDS-like codes that do not fall under the category of elliptic curve codes. By employing different techniques and considerations, they were able to construct these near MDS codes with unique properties. The works of Wu, Hyun, Lee \cite{WHL} and Han, Fan \cite{HF} demonstrated the effectiveness of the Roth-Lempel codes in constructing MDS codes with specific characteristics, such as the LCD property or the non-elliptic-curve type.

In this paper, we mainly consider the code $\mathcal C_2$ over the finite field $\gf(q)$   generated by the following matrix: 
\begin{equation}\label{eq:1.2} G_2=
\left( \begin{array}{ccccccc}
1 & 1& \ldots &  1 & 0 &0&0\\
\alpha_1& \alpha_2& \ldots &  \alpha_{n} & 0 &0&0\\
\vdots& \vdots& \ldots & \vdots & \vdots &\vdots&\vdots\\
\alpha_1^{k-4}& \alpha_2^{k-4}& \ldots &  \alpha_{n}^{k-4} & 0 &0&0\\
\alpha_1^{k-3}& \alpha_2^{k-3}& \ldots &  \alpha_{n}^{k-3} & 0 &0&1\\
\alpha_1^{k-2}& \alpha_2^{k-2}& \ldots &  \alpha_{n}^{k-2} & 0 &1&\tau\\
\alpha_1^{k-1}& \alpha_2^{k-1}& \ldots &  \alpha_{n}^{k-1} & 1 &\delta&\pi\\
\end{array} \right),
\end{equation}where $\alpha_{1},\ldots,\alpha_{n}$ are distinct elements in $\gf(q)$ and $ \delta,\tau,\pi \in \gf(q)$ and $4\le k + 1 \le n \le q$. 
Our motivation is the following open problem.  

\begin{problem}\label{pro} {\rm  When is the code $\mathcal{C}_2$ MDS,  or almost MDS,  or near MDS?}
\end{problem}

In this paper, we will provide solutions to Open Problem \ref{pro}. The rest of this  paper is organized as follows. In Section II, we will recall some basic concepts and required results. 
In Section III, we will show that the code $\mathcal C_2$ generated by $G_2$ is exactly an extended code of a  Roth-Lempel code, which helps us to prove that the code is of non-Reed-Solomon type. In Section IV, we will present a sufficient and necessary condition under which the code $\mathcal C_2$ is MDS. We also give MDS codes over small finite fields as examples. We determine  when the codes are almost MDS and near MDS in Section V. In Section VI, as applications of our main result, we will compute the covering radii of the duals of some Roth-Lempel codes and also  obtain an infinite family of (almost)  optimally extendable codes with dimension three. Finally, we conclude this paper in Section VII.
 \section{Preliminaries}

\subsection{Two known constructions of MDS codes}

There are many constructions of MDS codes in the literature.  In this subsection,  
we introduce two of them, which are needed later.  
We begin with the following well-known generalized Reed-Solomon codes.

 \begin{definition}{\rm \label{def1} 
 Let $\alpha_{1},\ldots,\alpha_{n}$ be distinct elements in $\gf(q)$  and $v_{1},\ldots,v_{n}$ be nonzero elements in $\gf(q)$. For $1\leq k \leq n$, the corresponding {\em generalized Reed-Solomon} (GRS) code over $\gf(q)$ is defined by
$$GR{S_k}({\boldsymbol{ \alpha}},{\bf v}): = \left\{({v_1}f({\alpha_1}), \ldots ,{v_{n }}f({\alpha_{n }})) \mid f(x) \in {\gf(q)}[x],\phantom{.} \deg(f(x)) < k \right\},$$ where $\boldsymbol{ \alpha}=(\alpha_{1},\alpha_{2},\ldots,\alpha_{n})\in \gf(q)^{n}$ and ${\bf v}=(v_{1},v_{2},\ldots,v_{n})\in {(\gf(q)^*)}^{n}$.
}
 \end{definition}

If $v_i=1$ for every $i=1, \ldots, n$, then $GR{S_k}({\boldsymbol{ \alpha}},{\bf 1})$ is called a {\it Reed-Solomon} (RS) code.
 It is well-known that a generalized Reed-Solomon code $GRS_k(\boldsymbol{\alpha},{\bf v})$ is an  $ {\left[ {n,k,n-k+1} \right] }$ MDS code.

Roth and Lempel \cite{RL} presented a new construction of MDS codes of non-Reed-Solomon type, which 
was introduced before.  Below we present a formal definition.

\begin{definition}\label{lem4} {\rm \cite{RL}
 Let $n$ and $k$ be two integers such that $4\le k+1\le n\le q$.  Let $\alpha_1, \ldots, \alpha_{n}$ be distinct elements of $\gf(q)$,  $\boldsymbol{\alpha}=(\alpha_{1},\alpha_{2},\ldots,\alpha_{n})$, and $\delta\in \gf(q)$. Then an $[n+2,k]$ {\em Roth-Lempel  code}
$RL({\boldsymbol{ \alpha}}, \delta,  k,n+2)$
over $\gf(q)$ is generated by the matrix  $G_1$ in Equation \eqref{eq:1.1}.
}
\end{definition}

The following lemma determines when a Roth-Lempel  code is MDS.

\begin{lemma}\label{lem: RL}{\rm \cite{RL} 
The Roth-Lempel  code $RL({\boldsymbol{ \alpha}}, \delta,  k,n+2)$ is  MDS  if and only if the set $\{\alpha_1, \ldots, \alpha_{n}\}$ is an $(n, k-1, \delta)$-set in $\gf(q)$, where a set $S\subseteq \gf(q)$ of size $n$ is called an $(n,k-1, \delta)$-set in $\gf(q)$ if there exists an element $\delta\in \gf(q)$ such that no $k-1$ elements of $S$ sum to $\delta$. 
}
\end{lemma}

\subsection{One class of special polynomials and generalized Vandermonde matrices}
Define $[l]:=\{1,2,\ldots,l\}$. The {\it multivariate homogeneous polynomial} of degree $r$ in $l$ variables $x_1, x_2,\ldots, x_{l}$, denoted by $P_{r,l}(x_1,\ldots,x_l)$, is defined recursively by
\begin{equation*}\label{homo1}
  P_{r,l}(x_1,x_2,\ldots,x_l)=x_lP_{r-1,l}
  (x_1,x_2,\ldots,x_l)+P_{r,l-1}(x_1,x_2,\ldots,x_{l-1});
\end{equation*}
\begin{equation*}\label{homo2}
  P_{r,1}(x_1)=x_1^r;~~~~P_{0,l}(x_1,x_2,\ldots,x_l)\equiv1;~~~~P_{-s,l}(x_1,x_2,\ldots,x_l)\equiv0,~s>0.
\end{equation*}
The polynomial $P_{r,l}(x_1,x_2,\ldots,x_l)$ is abbreviated as $P_{r,l}$. It is easy to verify that  \begin{equation}\label{eq:prl}
P_{r,l}=\sum\limits_{i_1+\cdots+i_l=r}\prod\limits_{j=1}^{l} x_j^{i_j},
\end{equation} where $i_1,\ldots,i_l$ are nonnegative integers.

An $m\times m$ matrix
$$V_{i_1,i_2,\ldots,i_{m-1}}(\alpha_1,\alpha_2,\ldots,\alpha_m)=\left(\begin{array}{cccc}
1&1&\cdots&1\\
\alpha_1^{i_1}&\alpha_2^{i_1}&\cdots&\alpha_m^{i_1}\\
\vdots&\vdots&\vdots&\vdots\\
\alpha_1^{^{i_{m-1}}}&\alpha_2^{^{i_{m-1}}}&\cdots&\alpha_m^{^{i_{m-1}}}
\end{array}\right)
$$
is called a {\it generalized Vandermonde matrix}, where  $\alpha_i\in \gf(q)$ are pairwise distinct. 
 Moreover, it is called a {\it Vandermonde matrix}, and written simply as $V(\alpha_1,\alpha_2,\ldots,\alpha_m)$, if 
 $i_j = j$ for  $1 \leq j \leq m-1$.

The following lemma establishes a useful connection between the Vandermonde determinant and the generalized Vandermonde determinant.

\begin{lemma}\label{gvand} {\rm \cite[Sec. 338]{Muir}\cite{HVDV}
Let $D_{k_1,\ldots,k_m}(x_1,\ldots,x_{m+1})$ denote the determinant of the matrix $$\left(\begin{array}{cccc}
P_{k_j-1,2}&P_{k_j-2,3}&\cdots&P_{k_j-m,m+1}
\end{array}\right)_{1\leq j\leq m},$$
where $(P_{k_j-1,2}~~P_{k_j-2,3}~~\cdots~~P_{k_j-m,m+1})$ denotes the $j$-th row of this matrix.
Then the determinant of $V_{k_1,\ldots,k_m}(x_1,\ldots,x_{m+1})$ is equal to
\begin{equation*}\label{didi}
  D_{k_1,\ldots,k_m}(x_1,\ldots,x_{m+1})\det(V(x_1,\ldots,x_{m+1})).
\end{equation*} }
\end{lemma}

 \subsection{A kind of extended codes of linear codes}

The Roth-Lempel code is a kind of extended code of the RS code, where the extended coordinates may or may not be linearly dependent on the coordinates in the original code.    
For a given linear code $\C$ of length $n$ over $\gf(q)$ and a nonzero vector $\bu$ in $\gf(q)^n$, Sun, Ding and Chen \cite{SDC}  recently defined another kind of extended linear code $\overline{\C}(\bu)$ of $\C$,  developed some general theory of the extended codes   $\overline{\C}(\bu)$   and 
studied  the extended codes  $\overline{\C}(\bu)$ of several families of linear codes, including cyclic codes, projective two-weight codes, nonbinary Hamming codes.

Let ${\bf u}=(u_1, u_2, \ldots, u_n)\in \gf(q)^n$ be any nonzero vector. In \cite{SDC},  any given $[n,k,d]$ code $\mathcal C$ over $\gf(q)$ can be extended into an $[n+1,k, \overline{d}]$ code $\overline{\mathcal C}({\bf u})$ over $\gf(q)$ as follows:
\begin{equation}\label{eq3}
\overline{\mathcal C}({\bf u})=\bigg\{  (c_1, \ldots, c_n,c_{n+1}): (c_1, \ldots, c_n)\in \mathcal C,c_{n+1}=\sum_{i=1}^nu_ic_i\bigg\},
\end{equation}
 and $\overline{d} = d$ or $\overline{d} = d +1$.  By definition,  $\overline{\C}(-\bone)$ is the classical extended code, where $\bone$ denotes the all-one vector of length $n$.  By  definition, the extended coordinate in 
 $\overline{\C}(\bu)$ must be linearly dependent on the $n$ coordinates in $\C$, where the dependency is defined by the vector $\bu$.  This is the difference between the two kinds of extended codes in the literature.  
 

 It is easy to check that the following lemma holds.
 
 \begin{lemma}\label{lem1}{\rm  \cite{SDC} 
 Let $\mathcal C$ be an $[n,k,d]$ linear code over $\gf(q)$ and $\bu \in \gf(q)^n$. If $\mathcal C $ has generator matrix $G$ and parity check matrix $H$, then the generator and parity check matrices for the extended code $\overline{\mathcal C}({\bf u})$ in \eqref{eq3} are $(G~ G{\bf u}^T)$ and \begin{equation*}\left(\begin{array}{cccc} H & {\bf 0}^T\\ {\bf u} & -1
\end{array}\right), \end{equation*}
where ${\bf u}^T$ denotes the transpose of ${\bf u}$.
}
\end{lemma}

The \emph{covering radius}  of a code $\mathcal C$, denoted by $\rho(\mathcal C)$, is the maximum distance from any vector in $\gf(q)^n$ to the nearest codeword in $\mathcal C$. A \emph{deep hole} is a vector achieving the covering radius.

 \begin{lemma} \label {thm6}{\rm  \cite[Theorem 6]{WDC} 
 Let $\mathcal C$ be an $[n,k]$ MDS code over $\gf(q)$.  Then for any  ${\bf u}\in \gf(q)^n$, the extended code  $\overline{\mathcal C}({\bf u})$ in \eqref{eq3} is MDS   if and only if $\rho(\mathcal C^{\bot})=k$ and  ${\bf u}$ is a deep hole of the dual code $\mathcal C^{\bot}$.
 }
\end{lemma}


Recall that a \emph{monomial matrix} is a square matrix which has exactly one nonzero entry in each row and each column \cite[Sections 1.6 and 1.7]{HP}. 

\begin{definition}{\rm 
Let $ C_1$ and $ C_2$ be two linear codes  of the same length over $\gf(q)$,  and let $M_1$ be a generator matrix of $ C_1$. Then
 $ C_1$ and $ C_2$ are {\em monomially equivalent} if
 there is a monomial matrix $N$ such that $M_1N$ is a generator matrix of $ C_2$.
}
 \end{definition}
 
 GRS codes are  monomially equivalent to RS codes. The following lemma is easy to check and plays an important role in our construction of MDS codes of non-Reed-Solomon type.
 
 \begin{lemma}\label{lem:eqv}{\rm
  If the extended code $\overline{\C}(\bu)$ is monomially equivalent to an RS code, then  the original code $\C$ is also monomially equivalent to an RS code.   
  }
 \end{lemma}
 
Moreover, in the whole paper, if a code is not monomially equivalent to an RS code, then we call it a code of {\em non-Reed-Solomon type}.

\section{Extended codes of Roth-Lempel codes}

In this section, we will show that the code $\mathcal C_2$ generated by the matrix in \eqref{eq:1.2} is   an extended code of a Roth-Lempel code.

\begin{theorem}\label{ext}{\rm  
The code  $\mathcal C_2=\overline{RL({\bf a},\delta, k,n+2)}({\bf u})$, where ${\bf u}=(u_1,  \ldots, u_{n+2})\in \gf(q)^{n+2}$ and for $1\le i\le n,$
$$u_i=   \frac{\alpha_i^{n+2-k}}{\prod_{1\leq j\leq n,j\neq i}(\alpha_i-\alpha_j)},$$  
 $$u_{n+1}=\pi-(\tau-\sum_{i=1}^n \alpha_i)\delta+\sum_{1\le i<j\le n}\alpha_i\alpha_j-(\sum_{i=1}^n \alpha_i)^2 \mbox{ and }u_{n+2}=\tau -\sum_{i=1}^n \alpha_i.$$
}
\end{theorem}

\begin{proof}  By Lemma \ref{lem1}, it suffices  to check that $G_{1}{\bf u}^T=(0, 0\ldots, 0,1, \tau,\pi)^T$. 

According to the Cramer Rule, the system of the equations
\begin{equation}\label{eq:3.2}
\left(\begin{array}{cccc} 1 &1&\ldots  &1\\ \alpha_1 &\alpha_2 &\ldots & \alpha_{n}\\ \vdots &\vdots &\ddots&\vdots \\ \alpha_1^{n-1} &\alpha_2^{n-1} &\ldots &\alpha^{n-1}_{n}\end{array}\right)\left(\begin{array}{c}x_1\\x_2\\ \vdots
\\ x_n\end{array}\right)=\left(\begin{array}{c}0\\ \vdots
\\ 0\\1 \end{array}\right)
\end{equation}
has a nonzero solution $(w_1,\ldots, w_{n})^T$, where 
\begin{equation}\label{eq:3.3}
{w_i} = \frac{1}{\prod_{1\leq j\leq n,j\neq i}(\alpha_i-\alpha_j)}
\end{equation} with $i=1,2,\ldots,n$.  From Equation \eqref{eq:3.2}, we only need to check that  $$(\alpha_1^{k-2}, \ldots, \alpha_n^{k-2},0,1){\bf u}^T=\sum_{i=1}^nw_i\alpha_i^n+u_{n+2}=\tau$$ and $$(\alpha_1^{k-1}, \ldots, \alpha_n^{k-1},1,\delta){\bf u}^T=\sum_{i=1}^nw_i\alpha_i^{n+1}+u_{n+1}+\delta u_{n+2}=\pi.$$

Let $m(x)=\prod_{j=1}^n(x-\alpha_j)=x^n+\sum_{j=0}^{n-1}a_jx^j$.  Note that $m(\alpha_i)=0$ and hence $\alpha_i^n=-\sum_{j=0}^{n-1}a_j\alpha_i^j$ for $1\le i\le n$.  By Equation \eqref{eq:3.2}, we have $$\sum_{i=1}^nw_i\alpha_i^n=-\sum_{j=0}^{n-1}a_j(\sum_{i=1}^{n}w_i\alpha_i^j)=-a_{n-1}=\sum_{i=1}^n\alpha_i$$ and $$\sum_{i=1}^nw_i\alpha_i^{n+1}=-\sum_{j=0}^{n-1}a_j(\sum_{i=1}^{n}w_i\alpha_i^{j+1})=-(a_{n-2}-a_{n-1}^2)=-\sum_{1\le i<j\le n}\alpha_i\alpha_j+(\sum_{i=1}^n \alpha_i)^2.$$
This completes the proof.
\end{proof}


In their work \cite{WHL}, the authors established a proof demonstrating that Roth-Lempel codes are not monomially equivalent to RS codes. Building upon this observation, it follows from Theorem \ref{ext} and Lemma \ref{lem:eqv} that the code $\mathcal{C}_2$, generated by the matrix  in Equation \eqref{eq:1.2}, is not monomially equivalent to an  RS code. Consequently, $\mathcal{C}_2$  is a code of  non-Reed-Solomon type.

\section{When is the code $\mathcal C_2$ MDS?}

In this section, our objective is to determine some necessary and sufficient conditions under which the code $\mathcal{C}_2$ is  MDS, which  provides a solution to Open Problem \ref{pro}. 
By thoroughly investigating these conditions, our aim is to gain a comprehensive understanding of the MDS property of $\mathcal{C}_2$.

\begin{theorem}\label{thm}{\rm  
The code $\mathcal C_2$ generated by the matrix in \eqref{eq:1.2} is an MDS code of non-Reed-Solomon type if and only if the following conditions hold:
\begin{enumerate}
\item  The set $\{\alpha_1, \ldots, \alpha_n\}$ is an $(n,k-1,\delta)$-set in $\gf(q)$.

\item  The set $\{\alpha_1, \ldots, \alpha_n\}$ is an $(n,k-2,\tau)$-set in $\gf(q)$.

\item  For any subset
$I$ with size $k-1$ of $\{1, \ldots, n\}$, $$\sum_{i\neq j\in I}\alpha_i\alpha_j+\pi\neq \tau (\sum_{i\in I} \alpha_i).$$

\item For any subset $I$ with size $k-2$ of $\{1, \ldots, n\}$, $$\pi+\delta(\sum_{i\in I}\alpha_i)\neq \tau\delta+\sum_{i\in I}\alpha_i^2+\sum_{ i\neq j\in I}\alpha_i\alpha_j.$$
\end{enumerate}
}
\end{theorem}

\begin{proof} The code $\mathcal C_2$ is MDS if and only if the submatrix  consisting of any $k$ columns from
$G_2$ in \eqref{eq:1.2} is nonsingular. Let ${\bf u}_1=(0,\ldots, 0,1)^T$, ${\bf u}_2=(0,\ldots,0,1, \delta)^T$, and ${\bf u}_3=(0,\ldots,0,1, \tau, \pi)^T$. Then we divide the proof into the following cases:
\begin{enumerate}
\item  Assume that the submatrix  consisting of $k$ columns from $G_2$ does not contain ${\bf u}_1, {\bf u}_2,$ and  ${\bf u}_3$. The matrix is a Vandermonde one and it is nonsingular of course.

\item Assume that the submatrix contains only ${\bf u}_1$ and other $k-1$ columns from $G_2$. It is also easy to check that the matrix is nonsingular.

\item Assume that the submatrix contains only ${\bf u}_2$ and other $k-1$ columns from $G_2$. Without loss of generality we need to compute the determinant of the matrix:
\begin{equation*}\label{eq10}
B_1=\left(\begin{array}{ccccc}
 1 &1 &\ldots  &1&0\\ 
 \alpha_1 &\alpha_2 &\ldots & \alpha_{k-1}&0\\ 
 \vdots &\vdots &\ddots&\vdots&\vdots \\ 
 \alpha_1^{k-3} &\alpha_2^{k-3} &\ldots & \alpha^{k-3}_{k-1}&0\\ 
 \alpha_1^{k-2} &\alpha_2^{k-2} &\ldots & \alpha^{k-2}_{k-1}&1\\
 \alpha_1^{k-1} &\alpha_2^{k-1} &\ldots & \alpha^{k-1}_{k-1}&\delta\\
\end{array}\right).\end{equation*} Let $A_i$ be the coefficient of $x^i$ in the polynomial $\mbox{det}(B(x))$, where
\begin{equation*}\label{eq11}
B(x)=\left(\begin{array}{ccccc}
 1 &1 &\ldots  &1&1\\ 
 \alpha_1 &\alpha_2 &\ldots & \alpha_{k-1}&x\\ 
 \vdots &\vdots &\ddots&\vdots&\vdots \\ 
 \alpha_1^{k-1} &\alpha_2^{k-1} &\ldots & \alpha^{k-1}_{k-1}&x^{k-1}\\
\end{array}\right).\end{equation*}
Then $\mbox{det}(B_1)=\delta A_{k-1}+A_{k-2}$. Note that $$\sum_{i=0}^{k-1}A_ix^i=\prod_{1\le i<j\le k-1}(\alpha_j-\alpha_i) \prod_{i=1}^{k-1}(x-\alpha_i).$$
Thus $A_{k-1}=\prod_{1\le i<j\le k-1}(\alpha_j-\alpha_i)\neq 0$, $A_{k-2}=-A_{k-1}\sum_{i=1}^{k-1}\alpha_i$ and $\mbox{det}(B_1)\neq 0$ if and only if $\sum_{i=1}^{k-1}\alpha_i\neq \delta.$ Hence,  the matrix is nonsingular if and only if the set $\{\alpha_1, \ldots, \alpha_n\}$ is an $(n,k-1,\delta)$-set in $\gf(q)$.

\item Assume that the submatrix contains only ${\bf u}_3$ and other $k-1$ columns from $G_2$. Without loss of generality we need to compute the determinant of the matrix:
\begin{equation*}\label{eq10}
B_2=\left(\begin{array}{ccccc}
 1 &1 &\ldots  &1&0\\ 
 \alpha_1 &\alpha_2 &\ldots & \alpha_{k-1}&0\\ 
 \vdots &\vdots &\ddots&\vdots&\vdots \\ 
  \alpha_1^{k-4} &\alpha_2^{k-4} &\ldots & \alpha^{k-4}_{k-1}&0\\ 
 \alpha_1^{k-3} &\alpha_2^{k-3} &\ldots & \alpha^{k-3}_{k-1}&1\\ 
 \alpha_1^{k-2} &\alpha_2^{k-2} &\ldots & \alpha^{k-2}_{k-1}&\tau\\
 \alpha_1^{k-1} &\alpha_2^{k-1} &\ldots & \alpha^{k-1}_{k-1}&\pi\\
\end{array}\right).\end{equation*} Similarly to Case (3), $\mbox{det}(B_2)=\pi A_{k-1}+\tau A_{k-2}+A_{k-3}$. Note that $$A_{k-3}=A_{k-1}\sum_{1\le i<j\le k-1}\alpha_i\alpha_j$$ and $\mbox{det}(B_2)\neq 0$ if and only if $$\sum_{1\le i<j\le k-1}\alpha_i\alpha_j+\pi-\tau \sum_{i=1}^{k-1}\alpha_i\neq 0.$$   Hence,  the matrix is nonsingular if and only if $$\sum_{i\neq j\in I}\alpha_i\alpha_j+\pi\neq \tau (\sum_{i\in I} \alpha_i)$$
for any subset
$I$ with size $k-1$ of $\{1, \ldots, n\}$.

\item Assume that the submatrix contains only ${\bf u}_1, {\bf u}_2$ and other $k-2$ columns from $G_2$. It is also easy to check that the matrix is nonsingular.

\item Assume that the submatrix contains only ${\bf u}_1, {\bf u}_3$ and other $k-2$ columns from $G_2$. Without loss of generality we need to compute the determinant of the matrix:
\begin{equation*}\label{eq10}B_3=
\left(\begin{array}{cccccc}
 1 &1 &\ldots  &1&0&0\\ 
 \alpha_1 &\alpha_2 &\ldots & \alpha_{k-2}&0&0\\ 
 \vdots &\vdots &\ddots&\vdots&\vdots &\vdots\\ 
 \alpha_1^{k-4} &\alpha_2^{k-4} &\ldots & \alpha^{k-4}_{k-2}&0&0\\ 
 \alpha_1^{k-3} &\alpha_2^{k-3} &\ldots & \alpha^{k-3}_{k-2}&1&0\\
  \alpha_1^{k-2} &\alpha_2^{k-2} &\ldots & \alpha^{k-2}_{k-2}&\tau&0\\
 \alpha_1^{k-1} &\alpha_2^{k-2} &\ldots & \alpha^{k-2}_{k-2}&\pi&1\\
\end{array}\right).\end{equation*} Hence, the matrix is nonsingular if and only if the set $\{\alpha_1, \ldots, \alpha_n\}$ is an $(n,k-2,\tau)$-set in $\gf(q)$.

\item Assume that the submatrix contains only ${\bf u}_2, {\bf u}_3$ and other $k-2$ columns from $G_2$. 
Without loss of generality we need to compute the determinant of the matrix:
\begin{equation*}\label{eq10}
B_4=\left(\begin{array}{cccccc}
 1 &1 &\ldots  &1&0&0\\ 
 \alpha_1 &\alpha_2 &\ldots & \alpha_{k-2}&0&0\\ 
 \vdots &\vdots &\ddots&\vdots&\vdots&\vdots \\ 
  \alpha_1^{k-4} &\alpha_2^{k-4} &\ldots & \alpha^{k-4}_{k-2}&0&0\\ 
 \alpha_1^{k-3} &\alpha_2^{k-3} &\ldots & \alpha^{k-3}_{k-2}&0&1\\ 
 \alpha_1^{k-2} &\alpha_2^{k-2} &\ldots & \alpha^{k-2}_{k-2}&1&\tau\\
 \alpha_1^{k-1} &\alpha_2^{k-1} &\ldots & \alpha^{k-1}_{k-2}&\delta&\pi\\
\end{array}\right).\end{equation*} 

Then \begin{eqnarray*}
&&\mbox{det}(B_4)=(-1)^{k-2+k}\left|\begin{array}{cccccc}
 1 &1 &\ldots  &1&0\\ 
 \alpha_1 &\alpha_2 &\ldots & \alpha_{k-2}&0\\ 
 \vdots &\vdots &\ddots&\vdots&\vdots \\ 
 \alpha_1^{k-4} &\alpha_2^{k-4} &\ldots & \alpha^{k-4}_{k-2}&0\\ 
 \alpha_1^{k-2} &\alpha_2^{k-2} &\ldots & \alpha^{k-2}_{k-2}&1\\
 \alpha_1^{k-1} &\alpha_2^{k-1} &\ldots & \alpha^{k-1}_{k-2}&\delta\\
\end{array}\right| +(\pi-\tau \delta)\prod_{1\le i<j\le k-2}(\alpha_j-\alpha_i)\\
&=&-\left|\begin{array}{cccccc}
 1 &1 &\ldots  &1\\ 
 \alpha_1 &\alpha_2 &\ldots & \alpha_{k-2}\\ 
 \vdots &\vdots &\ddots&\vdots \\ 
 \alpha_1^{k-4} &\alpha_2^{k-4} &\ldots & \alpha^{k-4}_{k-2}\\ 
 \alpha_1^{k-1} &\alpha_2^{k-1} &\ldots & \alpha^{k-1}_{k-2}\\
\end{array}\right|+\delta  \left|\begin{array}{cccccc}
 1 &1 &\ldots  &1\\ 
 \alpha_1 &\alpha_2 &\ldots & \alpha_{k-2}\\ 
 \vdots &\vdots &\ddots&\vdots \\ 
 \alpha_1^{k-4} &\alpha_2^{k-4} &\ldots & \alpha^{k-4}_{k-2}\\ 
 \alpha_1^{k-2} &\alpha_2^{k-2} &\ldots & \alpha^{k-2}_{k-2}\\
\end{array}\right|+(\pi-\tau \delta)\prod_{1\le i<j\le k-2}(\alpha_j-\alpha_i).
\end{eqnarray*}

From Case (3) we have  $$\left|\begin{array}{cccccc}
 1 &1 &\ldots  &1\\ 
 \alpha_1 &\alpha_2 &\ldots & \alpha_{k-2}\\ 
 \vdots &\vdots &\ddots&\vdots \\ 
 \alpha_1^{k-4} &\alpha_2^{k-4} &\ldots & \alpha^{k-4}_{k-2}\\ 
 \alpha_1^{k-2} &\alpha_2^{k-2} &\ldots & \alpha^{k-2}_{k-2}\\
\end{array}\right|=(\sum_{i=1}^{k-2}\alpha_i)\prod_{1\le i<j\le k-2}(\alpha_j-\alpha_i).$$
By Lemma \ref{gvand}, \begin{eqnarray*}
&&\left|\begin{array}{cccccc}
 1 &1 &\ldots  &1\\ 
 \alpha_1 &\alpha_2 &\ldots & \alpha_{k-2}\\ 
 \vdots &\vdots &\ddots&\vdots \\ 
 \alpha_1^{k-4} &\alpha_2^{k-4} &\ldots & \alpha^{k-4}_{k-2}\\ 
 \alpha_1^{k-1} &\alpha_2^{k-1} &\ldots & \alpha^{k-1}_{k-2}\\
\end{array}\right|\\
&=&\prod_{1\le i<j\le k-2}(\alpha_j-\alpha_i)\left|\begin{array}{cccccc}
 P_{0,2} &0 &\ldots  &0&0\\ 
P_{1,2} &P_{0,3} &\ldots &0& 0\\ 
 \vdots &\vdots &\ddots&\vdots \\ 
 P_{k-5,2} &P_{k-6,3} &\ldots &  P_{0,k-3}&0\\ 
  P_{k-2,2} &P_{k-3,3} &\ldots & P_{3,k-3}& P_{2,k-2}
\end{array}\right|.\\
&=&\prod_{1\le i<j\le k-2}(\alpha_j-\alpha_i)(\sum_{i=1}^{k-2}\alpha_i^2+\sum_{i\le i<j\le k-2}\alpha_i\alpha_j).
\end{eqnarray*}  The last equation holds because of  the fact  that $P_{0,l}=1$ and $$P_{r,l}=\sum\limits_{i_1+\cdots+i_l=r}\prod\limits_{j=1}^{l} x_j^{i_j},$$ where $i_1,\ldots,i_l$ are nonnegative integers.
Hence  $\mbox{det}(B_4)\neq 0$ if and only if $$\pi-\tau\delta-\sum_{i=1}^{k-2}\alpha_i^2-\sum_{1\le i<j\le k-2}\alpha_i\alpha_j-\delta(\sum_{i=1}^{k-2}\alpha_i)\neq 0.$$ Consequently, the matrix is nonsingular if and only if $$\pi-\tau\delta-\sum_{i\in I}\alpha_i^2-\sum_{ i\neq j\in I}\alpha_i\alpha_j+\delta(\sum_{i\in I}\alpha_i)\neq 0$$
for any subset
$I$ with size $k-2$ of $\{1, \ldots, n\}$.

\item Assume that the submatrix contains ${\bf u}_1,{\bf u}_2, {\bf u}_3$ and other $k-3$ columns from $G_2$. It is also easy to check that the matrix is nonsingular.
\end{enumerate}
This completes the proof.
\end{proof}

We have agreed with the reader that $\prod_{ i\neq j\in I}\alpha_i\alpha_j=0$ if $|I|=1.$
We give the following examples, which are all confirmed  by Magma.
\begin{example}{\rm
Let $q=4$ and $k=n=3$. Let $w$ be a primitive element of the finite field $\gf(4)$ with $w^2+w+1=0$.
\begin{itemize}
\item Let $\boldsymbol{\alpha}=(0,1,w)$. From Theorem \ref{thm}, we have 
\begin{eqnarray*}
\begin{cases}  (1) \Longrightarrow \delta \neq 1,w,1+w,\\
(2) \Longrightarrow \tau \neq 0,1,w,\\
(3) \Longrightarrow  \pi\neq \tau, \pi\neq w\tau, w+\pi\neq \tau(1+w),\\
(4) \Longrightarrow  \pi-\tau\delta\neq 0, \pi-\tau\delta-1-\delta\neq 0,\pi-\tau\delta -w^2-\delta w\neq 0.\\
   \end{cases}
\end{eqnarray*} Hence we have $(\delta, \tau,\pi)=(0,1+w,w)$.

\item Let $\boldsymbol{\alpha}=(0,1,1+w)$. From Theorem \ref{thm}, we have 
\begin{eqnarray*}
\begin{cases}  (1) \Longrightarrow \delta \neq 1,w,1+w,\\
(2) \Longrightarrow \tau \neq 0,1,1+w,\\
(3) \Longrightarrow  \pi\neq \tau, \pi\neq (1+w)\tau, 1+w+\pi\neq \tau w,\\
(4) \Longrightarrow  \pi-\tau\delta\neq 0, \pi-\tau\delta-1-\delta\neq 0,\pi-\tau\delta -w-\delta (1+w)\neq 0.\\
   \end{cases}
\end{eqnarray*} Hence we have $(\delta, \tau,\pi)=(0,w,1+w)$.

\item Let $\boldsymbol{\alpha}=(0,w,1+w)$. From Theorem \ref{thm}, we have 
\begin{eqnarray*}
\begin{cases}  (1) \Longrightarrow \delta \neq 1,w,1+w,\\
(2) \Longrightarrow \tau \neq 0,w,1+w,\\
(3) \Longrightarrow  \pi\neq w\tau, \pi\neq (1+w)\tau, 1+\pi\neq \tau,\\
(4) \Longrightarrow  \pi-\tau\delta\neq 0, \pi-\tau\delta-w^2-w\delta\neq 0,\pi-\tau\delta -w-\delta (1+w)\neq 0.\\
   \end{cases}
\end{eqnarray*} Hence we have $(\delta, \tau,\pi)=(0,1,1)$.

\item Let $\boldsymbol{\alpha}=(1,w,1+w)$. From Theorem \ref{thm}, we have 
\begin{eqnarray*}
\begin{cases}  (1) \Longrightarrow \delta \neq 1,w,1+w,\\
(2) \Longrightarrow \tau \neq 1,w,1+w,\\
(3) \Longrightarrow  w+\pi\neq (1+w)\tau, 1+w+\pi\neq w\tau, 1+\pi\neq \tau ,\\
(4) \Longrightarrow  \pi-\tau\delta-1-\delta\neq 0, \pi-\tau\delta-w^2-w\delta\neq 0,\pi-\tau\delta -w-\delta (1+w)\neq 0.\\
   \end{cases}
\end{eqnarray*} Hence we have $(\delta, \tau,\pi)=(0,0,0)$.
\end{itemize}

In this specific case $(q, k, n)=(4, 3, 3)$, there are only four instances where the code $\C_2$ is MDS with parameters $[6,3,4]$ over $\gf(4)$. Remarkably, these four MDS codes are all equivalent, aligning with the findings of a Magma program. The single MDS code that falls under the non-Reed-Solomon type has the same parameters $[6,3,4]$ over $\gf(4)$, and its generator matrix is provided as follows:
\begin{equation*}\left(\begin{array}{cccccc} 
1&1&1&0&0&1\\ 
1&w&w^2&0&1&0\\ 
1&w^2&w&1&0&0\\ 
\end{array}\right).\end{equation*}
}
\end{example}

\begin{example}{\rm   
   Let $q=5$ and $k=n=3$, and $\boldsymbol{\alpha}=(1,2,3)$. Let  $(\delta, \tau,\pi)=(2,0,1).$ Then   
   $\C_2$ is an MDS code of non-Reed-Solomon type with parameters $[6,3,4]$ over $\gf(5)$. 
}
\end{example}

\begin{example}{\rm
 Let $q=7$ and $k=n=3$, and $\boldsymbol{\alpha}=(2,3,5)$. Let $(\delta, \tau,\pi)=(3,0,2). $ Then 
 $\C_2$ is  
 an MDS code with parameters $[6,3,4]$ over $\gf(7)$.  By Magma, there are  other 27 choices of 
 $(\delta, \tau,\pi)$ and some induced MDS codes $\C_2$ are not equivalent. 
 }
\end{example}

 \begin{example}{\rm
 Let $q=8$ and $\gamma$ be a primitive element of the finite field $\gf(8)$.
 
 (1)  If $\boldsymbol{\alpha}=(0,1,\gamma,\gamma^3)$, $n=4$, $k=3$ and $(\delta, \tau,\pi)=(\gamma^6,\gamma^5,\gamma^2) $,  then $\C_2$ is 
 an MDS code with parameters $[7,3,5]$ over $\gf(8)$. 

(2)  If $\boldsymbol{\alpha}=(0,1,\gamma,\gamma^3)$, $n=4$,  $k=4$ and $(\delta, \tau,\pi)=(\gamma^6,\gamma^6,0) $,  then $\C_2$ is 
 an MDS code with parameters $[7,4,4]$ over $\gf(8)$. 
 
 (3)  If $\boldsymbol{\alpha}=(0,1,\gamma,\gamma^2,\gamma^3)$, $n=5$,  and $k=3$ and $(\delta, \tau,\pi)=(0,\gamma^4,\gamma) $,  then $\C_2$ is 
 an MDS code with parameters $[8,3,6]$ over $\gf(8)$. 
 }
\end{example}

\section{When is the code $\mathcal C_2$ NMDS?}

In this section, our focus will be on searching for some conditions for the code $\mathcal{C}_2$ being AMDS 
or NMDS,  which  provides a solution to Open Problem \ref{pro}. 

\subsection{When is the dual code $\C_2^{\bot}$ AMDS? }

In this subsection, we look for conditions under which the dual code $\C_2^{\bot}$ is AMDS.
Since $G_2$ is a parity-check matrix of $\mathcal C_2^{\bot}$, we have the following result.

\begin{theorem}\label{thm: amds}  {\rm
The dual code $\mathcal C_2^{\bot}$  is AMDS if and only if  one of the following  holds. 

(1)  The set $\{\alpha_1, \ldots, \alpha_n\}$ is an $(n,k-2,\tau)$-set in $\gf(q)$ and  there exists a subset $I\subseteq \{1,\ldots, n\}$ with size $k-1$ such that $\sum_{i\in I}\alpha_i=\delta$ or $$\sum_{i\neq j\in I}\alpha_i\alpha_j+\pi= \tau (\sum_{i\in I} \alpha_i);$$ or  there exists a subset $J\subseteq \{1,\ldots, n\}$ with size $k-2$ such that $$\pi+\delta(\sum_{i\in J}\alpha_i)= \tau\delta+\sum_{i\in J}\alpha_i^2+\sum_{ i\neq j\in J}\alpha_i\alpha_j.$$

(2) For any subset
$J$ with size $k-2$ of $\{1, \ldots, n\}$, we have $$\pi\neq\sum_{i\in I} \alpha_i^2+\sum_{i\neq j\in I}\alpha_i\alpha_j$$ and  there exists a subset $I\subseteq \{1,\ldots, n\}$ with size $k-1$ such that $\sum_{i\in I}\alpha_i=\delta$ or $$\sum_{i\neq j\in I}\alpha_i\alpha_j+\pi= \tau (\sum_{i\in I} \alpha_i);$$ or  there exists a subset $J\subseteq \{1,\ldots, n\}$ with size $k-2$ such that $\sum_{i\in J} \alpha_i=\tau$ or $$\pi+\delta(\sum_{i\in J}\alpha_i)= \tau\delta+\sum_{i\in J}\alpha_i^2+\sum_{ i\neq j\in J}\alpha_i\alpha_j.$$
}
\end{theorem}

\begin{proof} The code $\mathcal C_2^{\bot}$ is  AMDS if and only if it has parameters $[n+3,n+3-k,k]$. By definition,  $\mathcal C_2^{\bot}$ has minimum distance $d=k$ if and only if the following hold:
\begin{enumerate}
\item  Any $k-1$ columns from $G_2$ in Equation \eqref{eq:1.2} are linearly independent.

\item There exists a  submatrix consisting of certain $k$ columns from $G_2$  in Equation \eqref{eq:1.2} is singular. 
\end{enumerate}

Let ${\bf u}_1=(0,\ldots, 0,1)^T$, ${\bf u}_2=(0,\ldots,0,1, \delta)^T$, and ${\bf u}_3=(0,\ldots,0,1, \tau, \pi)^T$. Then we divide the proof into the following cases. 

\begin{enumerate}
\item  Assume that the submatrix  consisting of $k-1$ columns from $G_2$ does not contain ${\bf u}_1, {\bf u}_2,$ and  ${\bf u}_3$. 
Deleting the last row of this submmatrix results in a subsubmatrix which is a Vandermonde subsubmartix.  
Hence, the $k-1$ columns of this submatrix are linearly independent.

\item Assume that the submatrix contains only ${\bf u}_1$ and other $k-2$ columns from $G_2$. It is also easy to check that the columns of this submatrix are also linearly independent.

\item Assume that the submatrix contains only ${\bf u}_2$ and other $k-2$ columns from $G_2$.  This 
submatrix has a Vandermonde subsubmartix.  Hence,  the columns of this submatrix are linearly independent.

\item Assume that the submatrix contains only ${\bf u}_3$ and other $k-2$ columns from $G_2$. Without loss of generality we need to compute the determinant of the matrix:
\begin{equation*}\label{eq10}
D_1=\left(\begin{array}{ccccc}
 1 &1 &\ldots  &1&0\\ 
 \alpha_1 &\alpha_2 &\ldots & \alpha_{k-2}&0\\ 
 \vdots &\vdots &\ddots&\vdots&\vdots \\ 
  \alpha_1^{k-4} &\alpha_2^{k-4} &\ldots & \alpha^{k-4}_{k-2}&0\\ 
 \alpha_1^{k-3} &\alpha_2^{k-3} &\ldots & \alpha^{k-3}_{k-2}&1\\ 
 \alpha_1^{k-2} &\alpha_2^{k-2} &\ldots & \alpha^{k-2}_{k-2}&\tau\\
 \alpha_1^{k-1} &\alpha_2^{k-1} &\ldots & \alpha^{k-1}_{k-2}&\pi\\
\end{array}\right).\end{equation*} 
By \eqref{eq:3.2}, $D_1(x_1,x_2,\ldots, x_{k-1})^T={\bf 0}$ has a nonzero solution $(w_1, \ldots, w_{k-2},-1)$ if and only if $\tau=\sum_{i=1}^{k-2}\alpha_i$ and $\pi=-\sum_{1\le i<j\le k-2}\alpha_i\alpha_j+(\sum_{i=1}^{k-2} \alpha_i)^2$. Hence those $k-1$ columns  are linearly independent if and only if $$\tau \neq \sum_{i\in I} \alpha_i \mbox{ or }\pi\neq \sum_{i\in I} \alpha_i^2+\sum_{i\neq j\in I}\alpha_i\alpha_j$$
for any subset
$I$ with size $k-2$ of $\{1, \ldots, n\}$.

\item Assume that the submatrix contains only ${\bf u}_1, {\bf u}_2$ and other $k-3$ columns from $G_2$. Deleting the $(k-2)$-th row, we can obtain a nonsingular subsubmartix.  Hence, the columns of the submatrix  are linearly independent.

\item Assume that the submatrix contains only ${\bf u}_1, {\bf u}_3$ and other $k-3$ columns from $G_2$. Deleting the $(k-1)$-th row, we can obtain a nonsingular subsubmartix.  Hence, the columns of the submatrix 
are linearly independent.

\item Assume that the submatrix contains only ${\bf u}_2, {\bf u}_3$ and other $k-3$ columns from $G_2$. 
Deleting the last row, we can obtain a nonsingular subsubmartix.  Hence, the columns of the submatrix 
are linearly independent.

\item Assume that the submatrix contains ${\bf u}_1,{\bf u}_2, {\bf u}_3$ and other $k-4$ columns from $G_2$.  Deleting the $(k-3)$-th row, we can obtain a nonsingular subsubmartix. Hence, the columns of the submatrix 
are linearly independent. 
\end{enumerate}

In a word,  any $k-1$ columns  from $G_2$ in Equation \eqref{eq:1.2}  are linearly independent if and only if $$\tau \neq \sum_{i\in I} \alpha_i \mbox{ or }\pi\neq\sum_{i\in I} \alpha_i^2+\sum_{i\neq j\in I}\alpha_i\alpha_j$$
for any subset
$I$ with size $k-2$ of $\{1, \ldots, n\}$. In this case, the dual code $\mathcal C_2^{\bot}$ has parameters $[n+3, n+3-k, \ge k]$. Then the desired result follows from Theorem \ref{thm}.
\end{proof}


\subsection{Parity-check matrix of $\mathcal C_2$}

In this subsection, we will give a parity-check matrix of   the code $\mathcal C_2$.

\begin{theorem}\label{thm: dual} {\rm
Suppose that $\delta, \tau, \pi \in \gf(q)$, ${w_i} = \frac{1}{\prod_{1\leq j\leq n,j\neq i}(\alpha_i-\alpha_j)}$ was given in \eqref{eq:3.3}  with $i=1,2,\ldots,n$, and  $$a=\sum_{i=1}^n\alpha_i\mbox{ and }b=-\sum_{1\le i<j\le n}\alpha_i\alpha_j+(\sum_{i=1}^n \alpha_i)^2.$$
Then 
\begin{equation}\label{eq:3.1} H=
\left( \begin{array}{ccccccc}
w_1 & w_2& \ldots &  w_n & 0 &0&0\\
w_1\alpha_1& w_2\alpha_2& \ldots &  w_n \alpha_{n} & 0 &0&0\\
\vdots& \vdots& \ldots & \vdots & \vdots &\vdots&\vdots\\
w_1\alpha_1^{n-k-1}& w_2\alpha_2^{n-k-1}& \ldots & w_n  \alpha_{n}^{n-k-1} & 0 &0&0\\
w_1\alpha_1^{n-k}& w_2\alpha_2^{n-k}& \ldots & w_n  \alpha_{n}^{n-k} & -1 &0&0\\
w_1\alpha_1^{n+1-k}&w_2 \alpha_2^{n+1-k}& \ldots & w_n  \alpha_{n}^{n+1-k} & \delta-a &-1&0\\
w_1\alpha_1^{n+2-k}&w_2 \alpha_2^{n+2-k}& \ldots &w_n   \alpha_{n}^{n+2-k} & \pi-b-\delta(\tau-a) &\tau-a&-1\\
\end{array} \right)
\end{equation} is a   parity-check matrix of $\mathcal C_2$.
}
\end{theorem}

\begin{proof} 
By Equation \eqref{eq:3.2}, we have \begin{equation}\label{eq:3.4}
a=\sum_{i=1}^nw_i\alpha_i^n=\sum_{i=1}^n\alpha_i
\end{equation} and 
\begin{eqnarray} \label{eq:3.5}
b&=&\sum_{i=1}^nw_i\alpha_i^{n+1}=-\sum_{1\le i<j\le n}\alpha_i\alpha_j+(\sum_{i=1}^n \alpha_i)^2.
\end{eqnarray} 
By Equation \eqref{eq:3.2}, $G_2H^T=0$ and hence $H$ in Equation \eqref{eq:3.1}  is a parity-check matrix of the code $\mathcal C_2$ from the facts that $G_2H^T=0$ and $H$ is of row-full-rank.  This completes the proof.
\end{proof}

\subsection{When is the code $\C_2^{}$ AMDS? }

Based on Theorem \ref{thm: dual} and the fact that $H$ is a parity-check matrix of $\mathcal{C}_2$, we obtain  the following theorem.

\begin{theorem}\label{thm: amds2}  {\rm
The code $\mathcal C_2$  is AMDS if and only if  one of the following  holds:

(1)  The set $\{\alpha_1, \ldots, \alpha_n\}$ is an $(n,k-1,\delta)$-set in $\gf(q)$ and  there exists a subset $I\subseteq \{1,\ldots, n\}$ with size $k-1$ such that  $$\sum_{i\neq j\in I}\alpha_i\alpha_j+\pi= \tau (\sum_{i\in I} \alpha_i);$$ or  there exists a subset $J\subseteq \{1,\ldots, n\}$ with size $k-2$ such that $\sum_{i\in J} \alpha_i=\tau$ or $$\pi+\delta(\sum_{i\in J}\alpha_i)= \tau\delta+\sum_{i\in J}\alpha_i^2+\sum_{ i\neq j\in J}\alpha_i\alpha_j.$$

(2) For any subset
$J$ with size $n+1-k$ of $\{1, \ldots, n\}$, we have $$ \pi-b-\delta(\tau-a)\neq\sum_{i\in J} \alpha_i^2+\sum_{i\neq j\in J}\alpha_i\alpha_j,$$ and  there exists a subset $I\subseteq \{1,\ldots, n\}$ with size $k-1$ such that $\sum_{i\in I}\alpha_i=\delta$ or $$\sum_{i\neq j\in I}\alpha_i\alpha_j+\pi= \tau (\sum_{i\in I} \alpha_i);$$ or  there exists a subset $J\subseteq \{1,\ldots, n\}$ with size $k-2$ such that $\sum_{i\in J} \alpha_i=\tau$ or $$\pi+\delta(\sum_{i\in J}\alpha_i)= \tau\delta+\sum_{i\in J}\alpha_i^2+\sum_{ i\neq j\in J}\alpha_i\alpha_j.$$
}
\end{theorem}

Combining Theorems \ref{thm}, \ref{thm: amds} and \ref{thm: amds2} yields the following corollary.

\begin{corollary}\label{cor:nmds}{\rm 
Suppose that the set $\{\alpha_1, \ldots, \alpha_n\}$ is an $(n,k-1, \delta )$-set and an $(n,k-2, \tau)$-set in $\gf (q)$. Then the code $\C_2$ is NMDS if and only if 
\begin{enumerate}
\item  There exists a subset
$I$ with size $k-1$ of $\{1, \ldots, n\}$ such that  $$\sum_{i\neq j\in I}\alpha_i\alpha_j+\pi= \tau (\sum_{i\in I} \alpha_i);\mbox{  or }$$ 

 \item There exists a subset
$I$ with size $k-2$ of $\{1, \ldots, n\}$  such that $$\pi+\delta(\sum_{i\in I}\alpha_i)= \tau\delta+\sum_{i\in I}\alpha_i^2+\sum_{ i\neq j\in I}\alpha_i\alpha_j.$$
\end{enumerate}
}
\end{corollary}

\section{Applications of the main results}

In this section, we will explore several applications of our main results derived in the previous sections.

\subsection{Covering radii of the dual of the Roth-Lempel codes}

In this subsection, our objective is to present the cover radii of the dual codes of selected Roth-Lempel codes.
It is known that the cover radius of an $[n,k]$ MDS code is $n-k$ or $n-1-k$.  By Lemma \ref{lem: RL}, the Roth-Lempel  code is  MDS  if and only if the set $\{\alpha_1, \ldots, \alpha_{n}\}$ is an $(n, k-1, \delta)$-set in $\gf(q)$. Theorem \ref{thm} tells us when the code $\C _2$ is MDS. Hence, a Roth-Lempel  code is MDS when a companioning code $\C _2$ is MDS.
By  Lemma \ref{thm6}, the following theorem holds.

\begin{theorem}\label{thm:cov}{\rm 
Suppose that the conditions in Theorem \ref{thm} are satisfied and $4\le k+1  \le n \le q$.
Then the covering radius $\rho(RL({\bf a},  \delta, k,n+2)^{\bot})=k$ and the vector $\bu$ in  Theorem \ref{ext} is a deep hole of  $RL({\bf a},\delta, k,n+2)^{\bot}$.
}
\end{theorem}

To the best of our knowledge, this is the first time that the covering radii and deep holes of dual codes of Roth-Lempel codes are determined. We have the following examples to illustrate Theorem \ref{thm:cov}, which are confirmed by Magma.

\begin{example}{\rm 
Let $n=4$ and $\gf(8)^*=\langle w \rangle $.  Suppose that ${\bf a}=(0,1,w,w^2)$, $(\delta, \tau, \pi)=(w^2,w^2, w^5)$, and $k=3$. Then the code $\C_2$ is MDS and the covering radius $\rho(RL({\bf a},\delta, k,n+2)^{\bot})=3$. 
}
\end{example}

\begin{example}{\rm 
Let $n=4$ and $\gf(9)^*=\langle w \rangle $.  Suppose that ${\bf a}=(0,1,w,w^2)$, $(\delta, \tau, \pi)=(w^6,w^5, 1)$, and $k=3$. Then the code $\C_2$ is MDS and the covering radius $\rho(RL({\bf a},\delta, k,n+2)^{\bot})=3$.
}
\end{example}

\begin{example}{\rm 
Let $n=4$ and $q=5$.  Suppose that ${\bf a}=(1,2,3,4)$ and $k=3$. By Magma, there are no vector $(\delta, \tau, \pi)\in \gf(5)^3$ such that the code $\C_2$ is MDS and  the covering radius $\rho(RL({\bf a},\delta, k,n+2)^{\bot})=3$ when $\delta =0$ and 2 otherwise.
}
\end{example}

\subsection{(Almost) Optimal extendable linear codes}

In this subsection, as an application of Theorem \ref{thm}, we will obtain some (almost) optimal extendable linear codes with dimension three.

\begin{definition} \label{defi33}{\rm  \cite[Definition 1]{CLM}
Let $\bD$ be an $[ n, k ]$ linear code over a ﬁnite ﬁeld $\gf(q)$ and $G$ a generator matrix of $\bD$. Let $\bD'$ be the $[ n + k, k ]$ linear code with generator matrix $[ G : I_k ]$, where $I_k$ is the $k \times k$ identity matrix. We say that $\bD$ is {\it optimally extendable} if $\bD'$ has the same dual distance $d^{\bot}$ as $\bD$. Let $t$ be a positive integer such that $t \le d^{\bot}$. We say that $\bD$ is $t$-extendable if the dual distance of $\bD'$ is at least $t$. If $t$ is close to $d^{\bot}$,  $\bD$ is called almost optimally extendable.
}
\end{definition}


Indeed, constructing (almost) optimally extendable codes is a challenging task, and the choice of the generating matrix $G$ plays a crucial role. It is important to consider the pair $(\bD, G)$ to fully characterize the optimally extendable property. In the realm of coding theory, Carlet, Li, and Mesnager \cite{CLM} made significant contributions by introducing (almost) optimally extendable codes. They utilized previous results on primitive irreducible cyclic codes and first-order Reed-Muller codes to construct these codes. By leveraging these constructions, they were able to present examples of (almost) optimally extendable codes, further advancing the field. Additionally, Quan, Yue, and Hu \cite{QYH} made notable contributions by constructing three classes of (almost) optimally extendable linear codes. They employed irreducible cyclic codes, maximum-distance-separable codes, and near maximum-distance-separable codes in their constructions. These codes exhibit the desired optimally extendable property or come close to achieving it, highlighting the potential of different code families in achieving this challenging objective. The works of Carlet, Li, and Mesnager \cite{CLM} and  Quan, Yue, and Hu \cite{QYH} shed light on the construction of (almost) optimally extendable codes and provide valuable insights into the utilization of various code families to achieve this property. 

As an application of our main results,  we have the following result. 

\begin{theorem}{\rm \label{thm:extcode}
Let $\alpha_{1},\ldots,\alpha_{n}$ be distinct elements in $\gf(q)$ and ${\boldsymbol{ \alpha}}=(\alpha_1, \ldots, \alpha_n)$.  

(1)  Then $GR{S_3}({\boldsymbol{ \alpha}},{\bf 1})$ is  an optimally extendable code if and only if  $\alpha_i\neq 0$ for all $1\le i\le n$ and $\alpha_i+\alpha_j\neq 0$ for $1\le i<j\le n$.
 
(2) Then $GR{S_3}({\boldsymbol{ \alpha}},{\bf 1})$ is  an almost optimally extendable code if and only if  $\alpha_i\neq 0$ for all $1\le i\le n$ and  $\alpha_i+\alpha_j= 0$ for some $1\le i<j\le n$.

}
\end{theorem}

\begin{proof} Note that the parameters of $\C=GR{S_3}({\boldsymbol{ \alpha}},{\bf 1})$ and its dual are $[n,3,n-2]$ and $[n,n-3,4]$, respectively. The generator matrix of $\C'$ in Definition \ref{defi33} is given as follows:
\begin{equation}\label{eq:G} G=
\left( \begin{array}{ccccccc}
1 & 1& \ldots &  1&1&0&0 \\
\alpha_1& \alpha_2& \ldots &  \alpha_{n}&0&1&0 \\
\alpha_1^{2}& \alpha_2^{2}& \ldots &  \alpha_{n}^{2} &0&0&1 \\
\end{array} \right).
\end{equation}
 Note that in Equation \eqref{eq:1.2} if  $k=3$ and $\delta=\tau=\pi=0$, then the induced $G_2$ could 
 also be a generator matrix of $\C'$  in the sense of permutation equivalence of linear codes.

By Theorem \ref{thm},  the code $\C'$ is MDS if and only if the following conditions hold:
\begin{enumerate}
\item  The set $\{\alpha_1, \ldots, \alpha_n\}$ is an $(n,2,0)$-set in $\gf(q)$.

\item  The set $\{\alpha_1, \ldots, \alpha_n\}$ is an $(n,1,0)$-set in $\gf(q)$.

\item  For any $1\le i\neq j\le n$, we have $\alpha_i\alpha_j \neq 0.$ 

\item For any $1\le i\le n$, we have $\alpha_i^2\neq 0.$
\end{enumerate}
Namely  $\C'$ is MDS with parameters $[n+3,3,n+1]$ if and only if $\alpha_i\neq 0$ for all $1\le i\le n$ and $\alpha_i+\alpha_j\neq 0$ for $1\le i<j\le n$. This proves (1).

Note that $d(\C^{\bot})=4$. The code $\C=GR{S_3}({\boldsymbol{ \alpha}},{\bf 1})$ is  an almost optimally extendable code if and only if  $d(\C'^{\bot} ) = 3$. In Equation \eqref{eq:G}, $G$ is a parity matrix of $\C'^{\bot} $. If $\alpha_i= 0$ for some $1\le i\le n$, then $d(\C'^{\bot} ) = 2$.
Next let $\alpha_i\neq 0$ for all $1\le i\le n$.
Then $d(\C'^{\bot} ) = 3$ if and only if there exists a singular submatrix $D$ of order 3 in $G$. Suppose that $\alpha_i+\alpha_j=0$ for some $1\le i<j\le n$. Then \begin{equation*} 
\left| \begin{array}{ccccccc}
1 & 1& 0 \\
\alpha_i& \alpha_j&1 \\
\alpha_i^{2}& \alpha_j^{2}&0 \\
\end{array} \right|=-(\alpha_j-\alpha_i)(\alpha_j+\alpha_i)=0.
\end{equation*}
This completes the proof.
\end{proof}

\section{Summary and concluding remarks}

The main contributions of this paper are summarized as follows:

\begin{itemize}

\item Theorem  \ref{ext} was established.  It says that the code $\mathcal{C}_2$ is exactly an extended code 
of this type $\overline{\C}(\bu)$, where $\C$ is a Roth-Lempel code. This result establishes a direct relationship between $\mathcal{C}_2$ and the corresponding Roth-Lempel code, providing valuable insights into their connection. 

\item  Some sufficient and necessary conditions for $\mathcal{C}_2$ being  an MDS code of non-Reed-Solomon  type were presented in Theorem \ref{thm}.


\item Some sufficient and necessary conditions for both $\mathcal{C}_2$ and its dual being almost MDS 
were derived and presented in Theorems \ref{thm: amds} and \ref{thm: amds2}. These results contribute to the study of the almost MDS property and shed light on the error-correcting capabilities of $\mathcal{C}_2$ and its dual.

\item  Some sufficient and necessary conditions for $\mathcal{C}_2$ being  near MDS were obtained and 
presented in Corollary \ref{cor:nmds}. 

\item The covering radii and deep holes of the dual codes of some Roth-Lempel codes were determined in Theorem \ref{thm:cov}. 

\item An infinite family of (almost) optimally extendable codes with dimension three was given in Theorem \ref{thm:extcode}. 
\end{itemize}


The study of extended codes of linear codes over finite fields is an interesting and hard problem, and further exploration in this direction would be  valuable.


\end{document}